\documentclass[11pt, reqno, oneside]{article}
\usepackage{geometry} 
\usepackage[affil-it]{authblk}

\usepackage{amsmath, amsthm, amsfonts, amssymb, amsbsy, bigstrut, graphicx, enumerate,  upref, longtable, comment, booktabs, array, caption, subcaption, multirow, bigdelim}

\usepackage[breaklinks]{hyperref}

\usepackage[numbers, sort&compress]{natbib}

\newcommand{\ep}{\epsilon}

\newcommand{\tf}{\tilde{f}}

\newtheorem{thm}{Theorem}[section]
\newtheorem{lmm}[thm]{Lemma}
\newtheorem{cor}[thm]{Corollary}
\newtheorem{prop}[thm]{Proposition}

\theoremstyle{definition}

\newtheorem{ex}[thm]{Example}

\newcommand{\ee}{\mathbb{E}}

\newcommand{\pp}{\mathbb{P}}

\newcommand{\rr}{\mathbb{R}}

\newcommand{\var}{\mathrm{Var}}

\numberwithin{equation}{section}

\newcommand{\tg}{\tilde{g}}

\usepackage[utf8]{inputenc}
\usepackage{amsmath}
\usepackage{amsfonts}
\usepackage{bbm}
\usepackage{mathtools}
\usepackage{tikz}
\usetikzlibrary{decorations.pathreplacing}
\usepackage{amsthm}
\usepackage[english]{babel}
\usepackage{fancyhdr}
\usepackage{amssymb}
\usepackage{enumitem}
\usepackage{qtree}
\usepackage{mathrsfs}

\newcommand{\E}{\mathbb{E}}

\renewcommand{\P}{\mathbb{P}}

\renewcommand{\tilde}{\widetilde}

\begin{document}
\title{Estimating large causal polytrees from small samples}
\author{Sourav Chatterjee\thanks{Department of Statistics, Stanford University, Stanford, USA. Email: \href{mailto:souravc@stanford.edu}{\tt souravc@stanford.edu}. 
}}
\affil{Stanford University}
\author{Mathukumalli Vidyasagar\thanks{Indian Institute of Technology, Hyderabad, India. 
Email: \href{m.vidyasagar@ee.iith.ac.in}{\tt m.vidyasagar@ee.iith.ac.in}}}
\affil{Indian Institute of Technology, Hyderabad}



\maketitle

\begin{center}
\emph{In memory of Professor K.~R.~Parthasarathy}
\end{center}

\begin{abstract}
We consider the problem of estimating a large causal polytree from a relatively small i.i.d.~sample. This is motivated by the problem of determining causal structure when the number of variables is very large compared to the sample size, such as in gene regulatory networks. We give an algorithm that recovers the tree with high accuracy in such settings. The algorithm works under essentially no distributional or modeling assumptions other than some mild non-degeneracy conditions.
\end{abstract}


\maketitle

\section{Introduction}\label{intro}
The problem of estimating causal structure from data is a central problem of causal inference. One of the earliest attempts at reconstructing causal structures, under the assumption that the underlying graph is a tree (such structures are called causal polytrees), was due to Rebane and Pearl~\cite{rp87}, who repurposed an old algorithm of Chow and Liu~\cite{chowliu68} to give a method for consistent estimation of causal polytrees (a term that was coined in \cite{rp87}). 

The Rebane--Pearl approach has several drawbacks in the modern context. First, it is based on mutual information, just like the original Chow--Liu algorithm. Estimating mutual information from data is notoriously time-consuming (see \cite{chatterjee21} for some numbers), and moreover, requires special assumptions on the distribution of the data. Second, it is not clear if the algorithm works in modern problems where the number of variables is far greater than the sample size. Such examples arise routinely in gene regulatory networks, where we may have $n$ as small as a few hundred and $p$ as large as many thousands (see \cite{singhetal18} for some examples). 

The same problems persist in other popular classical algorithms for causal structure recovery (that go beyond trees), such as the PC algorithm~\cite{sg91} and the IC algorithm~\cite{pv95}.  In the last two decades, several attempts have been made to design algorithms that work when the sample size is relatively small compared to the number of variables. For example, variants of the PC algorithm were shown to work when the number of variables grows polynomially with sample size~\cite{kb07, mkb09, cm14}. These works are limited by the assumption that the variables are jointly Gaussian. A generalization to the case where the joint distribution is transformable to Gaussian using copulas was proposed in~\cite{hd13}. 

The above algorithms all use conditional independence testing of some form. Due to the well-known difficulties of conditional independence testing in the absence of stringent distributional assumptions~\cite{shahpeters20}, there is a different body of work that uses `scored-based methods' that do not need conditional independence testing. Score-based methods in the Gaussian setting have been proposed in~\cite{chickering02, vb13, lb14}. The problem becomes less tractable in the absence of Gaussianity; solutions under various other structural assumptions have been proposed by various authors --- see, e.g.,~\cite{shimizu06, hoyer08, peters14, pb14, buhlmann14, nb16}. Methods that combine score-based methods with conditional independence testing have also been proposed, e.g., in \cite{tsamardinos06, nandy18}. 

Although score-based algorithms avoid conditional independence testing, they have their own problems. For example, score-based algorithms are computationally rather expensive, to the extent that it is hard to use them when the number of variables is very large. (For some ideas about speeding up the calculations, see~\cite{bs13, zheng}.) Moreover, they often require strong assumptions on the distribution of the data.

To summarize the discussion, there is now a multitude of different algorithms for causal structure recovery, but they all suffer from at least one of the following problems: (a) They are not fully nonparametric, often requiring strong distributional or modeling assumptions. (b) They are computationally expensive and consequently hard to implement if the number of variables is large. (c) It is not clear --- either theoretically or practically --- whether they truly work well when the number of variables is large compared to the sample size, even if the computational issue can be tackled.  

Very recently, researchers have started paying attention to the above deficiencies of the algorithms proposed in the literature. Fully nonparametric methods have been proposed in \cite{gao1, gao2, atb21} with finite sample guarantees. These methods seem to work well when the number of variables is comparable to the sample size, as do some recent parametric approaches, such as~\cite{louetal21}. 

In this paper, we propose a fully nonparametric algorithm that accurately recovers a causal polytree from the data even if the number of variables is much larger than the sample size. (Theoretically, we only need that the sample size is logarithmic in the number of variables.) The efficacy of the algorithm is demonstrated through theoretical results and simulated examples. The main deficiency of our approach is that it is applicable only when the causal structure is a tree.  

The paper is organized as follows. We begin with an introduction to directed acyclic graphs and their relevance in causal networks in Section \ref{dagsec}. The first part of our algorithm, which estimates the skeleton, is presented in Section \ref{algo}. The theoretical guarantee for the skeleton recovery algorithm is given in Section \ref{result}. Section \ref{directionsec} presents the second part of our algorithm, which recovers the directionalities of the arrows in the estimated skeleton. The theoretical guarantee for this part is given in Section \ref{directionresult}. Simulated examples are in Section \ref{simul}. The remaining sections are devoted to proofs.

\medskip
\noindent {\bf Acknowledgements.} 
S.C.'s research was partially supported by NSF grants DMS-1855484 and DMS-2113242. S.C.~thanks Mona Azadkia for helpful comments. M.V.'s research was supported by the Science and Engineering Research Board (SERB), Government of India.

\section{Directed acyclic graphs}\label{dagsec}
Let $X = (X_i)_{i\in V}$ be a finite collection of real-valued  random variables. A directed acyclic graph $G$ with vertex set $V$ is a directed graph with edges connecting elements of $V$, which contains no cycles. Given a directed graph $G$, each $i\in V$ has a well-defined set of `parent nodes' with arrowing coming into $i$. Let $p(i)$ denote this set of parent nodes, which may be empty. We say that $G$ describes the dependency structure of $X$ if there is some measure $\mu$ on $\rr$ such that $X$ has a probability density with respect to the product measure induced by $\mu$ on $\rr^V$,  which can be written as 
\begin{align}\label{density}
f(x) = \prod_{i\in V} f_i(x_i|(x_j)_{j\in p(i)}),
\end{align}
where $f_i$ is the conditional density of $X_i$ given $(X_j)_{j\in p(i)}$. If $p(i)$ is empty, $f_i$ is simply the marginal density of $X_i$. This is an example of a graphical model, sometimes called a `Bayesian network'. 

If we remove the directionalities of the edges, the resulting undirected graph is called the `skeleton' of $G$. If the skeleton is a tree, then $G$ is called a `causal polytree'.

Note that neither the DAG $G$, nor the decomposition \eqref{density} of the density $f$, is uniquely determined by the distribution of~$X$. If we specify that the skeleton of $G$ is a tree, then under mild conditions (e.g., those in Theorem \ref{mainresult} below), the skeleton is uniquely determined. However, even then, the directionalities of the arrows may not be uniquely determined by the law of $X$. This is nicely demonstrated by a class of graphs that we will call `outgoing causal polytrees'. We will say that a causal polytree $G$ is outgoing if every vertex is incident to at most one incoming edge. Under this condition, it is easy to see that there is a unique `root' which has no incoming edges, and given the root and the skeleton, the directionalities of the edges are fully determined. It is not hard to show that if $G$ describes the dependency structure on $X$, and $G'$ is an outgoing causal polytree with the same skeleton but a different root, then $G'$ also describes the dependency structure of $X$. For example, if $X= (X_1,X_2,X_3)$, and the dependency structure of $X$ is described by the outgoing causal polytree $1\to 2\to 3$, then it is also described by $1\leftarrow 2 \to 3$ and $1\leftarrow 2\leftarrow 3$, which are the two outgoing trees with the same skeleton as the first one but with roots $2$ and $3$, respectively.

\section{Algorithm for recovering the skeleton}\label{algo}
Our algorithm is divided into two parts. The first part, presented in this section, is for recovering the skeleton. Both parts make use of a coefficient of correlation proposed recently in \cite{chatterjee21}. This is defined as follows. Let $(x_1,y_1),\ldots,(x_n,y_n)$ be $n$ pairs of real numbers, where $n\ge 2$. Let $\pi$ be a permutation of $1,\ldots, n$ such that $x_{\pi(1)}\le x_{\pi(2)}\le \cdots \le x_{\pi(n)}$. If there is more than one such permutation, choose one uniformly at random. For each $i$, let $r_i$ be the number of $j$ such that $y_j\le y_{\pi(i)}$ and let $l_i$ be the number of $j$ such that $y_j\ge y_{\pi(i)}$. Then define the $\xi$-coefficient between the $x_i$'s and the $y_i$'s as 
\begin{align}\label{xidef}
\xi_n := 1-\frac{n\sum_{i=1}^{n-1}|r_{i+1}-r_i|}{2\sum_{i=1}^n l_i(n-l_i)}. 
\end{align}
Note that the $\xi$-coefficient is not symmetric --- that is, the $\xi$-coefficient between the $y_i$'s and the $x_i$'s is the not the same as the $\xi$-coefficient between the $x_i$'s and the $y_i$'s. Another thing to note is that the denominator on the right side is equal to zero if and only if all the $y_i$'s are equal. In this case, $\xi_n$ is left undefined in \cite{chatterjee21}; but in this paper we define $\xi_n$ to be $0$ if this happens.

Let $X = (X_i)_{i\in V}$ be  a finite collection of random variables, whose dependency structure is described by a causal polytree $G$. Let $T$ be the skeleton of $G$. 
Our data consists of $n$ i.i.d.~copies  $X^1,\ldots,X^n$ of the collection $X$. We produce an estimate $T_n$ of $T$ in the following two steps:
\begin{itemize}
\item[Step 1.] Let $X^m_i$ denote the $i^{\textup{th}}$ coordinate of $X^m$. For each distinct $i,j\in V$, let $\xi^n_{ij}$ denote the $\xi$-coefficient between $(X_i^m)_{1\le m\le n}$ and $(X_j^m)_{1\le m\le n}$. Define a subgraph $G_n$ of the complete graph on $V$ as follows. For each distinct pair of vertices $i,j\in V$, keep the undirected edge $(i,j)$ in $G_n$ if and only if there does not exist $k\in V\setminus\{i,j\}$ such that 
\begin{align}\label{mainxicond}
\xi^n_{ki} \ge  \xi^n_{ji} \ \text{ and } \ \xi^n_{kj} \ge  \xi^n_{ij}.
\end{align}
\item[Step 2.] For each undirected edge $(i,j)$ in $G_n$, let $w_{ij}^n := \min\{\xi_{ij}^n, \xi_{ji}^n\}$ be the weight of $(i,j)$. Let $T_n$ be the maximal weighted spanning forest (MWSF) of $G_n$ with these edge-weights. That is, $T_n$ maximizes the sum of edge-weights among all spanning forests of $G_n$. If there is more than one MWSF, choose one according to some arbitrary rule. This $T_n$ is our estimate of $T$.
\end{itemize}
The first step of the above algorithm is inspired by a proposal from \cite{singhetal18}, where a coefficient called the `$\phi$-mixing coefficient'~\cite{av14} is used with similar intent, instead of the $\xi$-coefficient. The second step is inspired by the Chow--Liu algorithm~\cite{chowliu68} mentioned earlier, with mutual information replaced by the $w$-weights defined above.

\section{Theoretical guarantee for skeleton recovery}\label{result}
To state our theoretical guarantee for the skeleton recovery algorithm proposed in the previous section, we need a small amount of preparation. First, recall that the maximal correlation coefficient $R(X,Y)$ between two random variables $X$ and $Y$~\cite{gebelein41, renyi59} is defined as the supremum of the Pearson correlation between $f(X)$ and $g(Y)$ over all measurable functions $f$ and $g$ such that $f(X)$ and $g(Y)$ have finite and nonzero variance. It is easy to see that $R(X,Y)\in [0,1]$. Note that if $X$ is a constant, then there will not exist an $f$ as above. In this case, we define $R(X,Y)=0$. Similarly, we let $R(X,Y)=0$ if $Y$ is a constant. 

Next, for a random variable $X$, we define a quantity $\alpha(X)$ that measures the `degree to which $X$ is not a constant', as $\alpha(X) := \pp(X\ne X')$ where $X'$ is an independent copy of $X$. We will refer to it as the `$\alpha$-coefficient' of $X$. Note that $\alpha(X)>0$ if and only if $X$ is not a constant.

Lastly, recall that if $(X_1,Y_1),(X_2,Y_2),\ldots$ are i.i.d.~copies of a pair of random variables $(X,Y)$, where $Y$ is non-constant, and $\xi_n$ is the $\xi$-correlation coefficient between $X_1,\ldots,X_n$ and $Y_1,\ldots,Y_n$, then one of the main results of \cite{chatterjee21} is that as $n\to\infty$, $\xi_n$ converges in probability to the limit
\begin{align}\label{xicordef}
\xi(X,Y) = \frac{\int \var(\ee(1_{\{Y\ge t\}}|X)) d\mu(t)}{\int \var(1_{\{Y\ge t\}}) d\mu(t)}
\end{align}
where $\mu$ is the law of $Y$. Moreover, $\xi(X,Y)$ is always in $[0,1]$, $\xi(X,Y)=0$ if and only if $X$ and $Y$ are independent, and $\xi(X,Y) = 1$ if and only if $Y$ is equal to a measurable function of $X$ with probability one. The only assumption required for this result is that $Y$ is not a constant. We will refer to $\xi(X,Y)$ as the `$\xi$-correlation' between $X$ and $Y$. In this article, we will prove the auxiliary result that $\ee(\xi_n)$ converges to $\xi(X,Y)$ as $n\to \infty$. This does not directly follow from the results of \cite{chatterjee21}; a small additional argument is needed (see Corollary \ref{xicor}). In \cite{chatterjee21}, $\xi(X,Y)$ is left undefined if $Y$ is a constant; in this article, we define $\xi(X,Y)$ to be $0$ if $Y$ is a constant.

Let us now return to the setting of Section \ref{algo}. The following theorem shows, in essence, that $T_n = T$ with high probability if no $X_i$ is a constant, no $X_i$ and $X_j$ are independent when $i$ and $j$ are neighbors in $T$, no $X_i$ and $X_j$ have maximal correlation $1$ when $i$ and $j$ are neighbors in $T$, and $n\gg \log p$.
\begin{thm}\label{mainresult}
Let $X = (X_i)_{i\in V}$ be a finite collection of random variables with a causal polytree skeleton $T$, as defined at the beginning of Section~\ref{intro}, and let $p:=|V|$. Let $T_n$ be the estimate of $T$ based on a sample of $n$ i.i.d.~copies of $X$, as defined in Section~\ref{algo}. For each $i$ and $j$, let $\xi_{ij}$ be the $\xi$-correlation coefficient between $X_i$ and $X_j$, let $R_{ij}$ be the maximal correlation coefficient between $X_i$ and $X_j$, and let $\alpha_i$ be the $\alpha$-coefficient of $X_i$, as defined above. Suppose that there exists $\delta\in (0,1)$ such that: 
\begin{enumerate}
\item $\xi_{ij}\ge \delta$ whenever $(i,j)$ is an edge in $T$.
\item $R_{ij} \le 1-\delta$ whenever $(i,j)$ is an edge in $T$.
\item $\alpha_i \ge \delta$ for all $i\in V$. 
\end{enumerate}
Furthermore, suppose that $n$ is so large that $|\ee(\xi_{ij}^n) - \xi_{ij}| \le \delta^2/8$ for all distinct $i,j\in V$. 
Then there exist positive constants $C_1(\delta)$ and $C_2(\delta)$ depending only on $\delta$, such that
\[
\pp(T_n \ne T) \le C_1(\delta)p^2 e^{-C_2(\delta) n}. 
\]
\end{thm}
We remark that the lower bound on $n$ required in the above result is likely to depend only on $\delta$ in non-pathological situations, because it depends only on the joint distributions on $(X_i,X_j)$ over all edges $(i,j)$ in $T$. If none of these joint distributions are too `weird', the condition $|\ee(\xi_{ij}^n) - \xi_{ij}| \le \delta^2/8$ is likely to be satisfied for all $n$ exceeding some threshold depending only on $\delta$.

Another remark is that the theorem shows that if for all edges $(i,j)$ in $T$, $R_{ij}<1$ and $X_i$ and $X_j$ are dependent, and each $X_i$ is non-constant, then $T$ must be the unique causal polytree skeleton.

\section{Algorithm for recovering directionalities}\label{directionsec}
In this section we present the second part of our algorithm, which recovers the directionalities of the edges after estimating the skeleton. We will need the following conditional dependence coefficient proposed in \cite{azadkiachatterjee21}. Let $X$, $Y$ and $Z$ be three real-valued random variables, and let $(X_1,Y_1,Z_1),\ldots,(X_n,Y_n,Z_n)$ be i.i.d.~copies of $(X,Y, Z)$, where $n\ge 2$. For each $i$, let $N(i)$ be the index $j$ such that $X_j$ is the nearest neighbor of $X_i$, where ties are broken uniformly at random. Let $M(i)$ be the index $j$ such that $(X_j, Z_j)$ is the nearest neighbor of $(X_i, Z_i)$ in $\rr^2$, again with ties broken uniformly at random. Let $R_i$ be the rank of $Y_i$, that is, the number of $j$ such that $Y_j\le Y_i$. Define 
\[
\tau_n(Y, Z|X) := \frac{\frac{1}{n^2}\sum_{i=1}^n (\min\{R_i, R_{M(i)}\} - \min\{R_i, R_{N(i)}\})}{\frac{1}{n^2} \sum_{i=1}^n (R_i - \min\{R_i, R_{N(i)}\})},
\]
interpreting it as zero if the denominator is zero. This is the statistic $T_n$ defined in \cite{azadkiachatterjee21}, but we call it $\tau_n$ here so that there is no confusion with the tree $T_n$. Now, for our collection $X= (X_i)_{i\in V}$, define
\[
\tau_{kji}^n := \tau_n(X_k, X_j|X_i).
\]
Here is the proposed algorithm for recovering the directionalities of the edges in  the estimated skeleton: 
\begin{itemize}
\item[Step 0.] Estimate the skeleton using the algorithm proposed in Section \ref{algo}. The rest of the algorithm is for assigning directionalities to the edges of this estimated skeleton.
\item[Step 1.] Inspect each vertex $i\in V$, in some order that is not dependent on the data. Consider two cases:
\begin{itemize}
\item[Case 1.] If no incoming edge into $i$ has yet been detected, then for each pair of neighbors $(j,k)$ of $i$ in the estimated skeleton (sequentially in some order that is not dependent on the data), check  whether $\tau_{kji}^n \ge \xi_{jk}^n$. If this is found to be  true for some $(j,k)$, then declare that the edges $(j,i)$ and $(k,i)$ are both directed towards $i$ (unless determined otherwise in a previous step), and move on to the next $i$. If this is not true, do nothing and move on to the next $i$. 
\item[Case 2.] If we have already determined that some neighbor $j$ of $i$ in the estimated skeleton has an edge directed towards $i$, then fix some such $j$ and check for every neighbor $k$ of $i$ (other than $j$) whether $\tau_{kji}^n \ge \xi_{jk}^n$. For those $k$ for which this holds, declare that the edge $(k,i)$ points towards $i$ (unless determined otherwise in a previous step). For those $k$ for which this does not hold, declare that the edge $(k,i)$ points towards $k$ (unless determined otherwise in a previous step). 
\end{itemize}
\item[Step 2.] Repeat Step 1 until a stage is reached where no extra directionalities are detected upon executing Step~1.
\item[Step 3.] Inspect each vertex $i\in V$, in some order that is not dependent on the data. If there is at least one edge directed towards $i$, assign outgoing directionality to any edge incident to $i$ whose directionality has not yet been assigned. Otherwise, move on to the next $i$.
\item[Step 4.] Repeat Step 3 until no further directionalities are detected.
\item[Step 5.] If there are some edges in the skeleton whose directionalities have remained undecided after the execution of the above steps, then declare their directions to be the same as in the outgoing tree with an arbitrarily chosen vertex as the root. 
\end{itemize}

\section{R package}
An R package for implementing the above algorithms, called `PolyTree', is now available on CRAN~\cite{package}.

\section{Theoretical guarantee for recovering directionalities}\label{directionresult}
We now present our theoretical guarantee for recovery of the full causal structure of $X$. 
Take any $i,j,k$ such that $i$ is a neighbor of both $j$ and $k$ in the skeleton $T$. Let 
\[
\tau_{kji} := \tau(X_k, X_j|X_i),
\]
where $\tau$ is the same as the population statistic $T$ defined in \cite{azadkiachatterjee21}, that is, for any three random variables $X$, $Y$, and $Z$,
\[
\tau(Y,Z|X) = \frac{\int \E(\var(\P(Y\ge t|Z,X)|X)) d\mu(t)}{\int \E(\var(1_{\{Y\ge t\}}|X)) d\mu(t)},
\]
where $\mu$ is the law of $Y$. Let $q_{kji}^n$ and $s_{ji}^n$ denote the numerator and denominator in the definition of $\tau_{kji}^n$. Let $q_{kji}$ and $s_{ji}$ denote the almost sure limits of $q_{kji}^n$ and $s_{ji}^n$ as $n\to \infty$. By \cite[Theorem 9.1]{azadkiachatterjee21}, these limits exist and are non-random. 
\begin{thm}\label{directionthm}
Let all assumptions of Theorem \ref{mainresult} be valid. Additionally, we make the following assumptions, with $\delta$ being the same as in Theorem \ref{mainresult}:
\begin{enumerate}
\item Take any $i,j,k$ such that $i$ is a neighbor of both $j$ and $k$ in the tree $T$. If the arrows are both directed towards $i$ in some valid DAG, then $\tau_{kji} \ge \delta$. If this fails in some valid DAG, then $\xi_{jk}\ge \delta$.
\item For any two neighboring vertices  $i$ and $j$ in $T$, the $s_{ij}\ge \delta$.
\item The sample size $n$ is so large that $|\ee(q_{kji}^n) - q_{kji}| \le \delta^3/8$ and $|\ee(s_{ji}^n) - s_{ji}|\le \delta^3/8$ for all distinct $i,j,k\in V$.
\end{enumerate}
Then there there positive constants $C_1(\delta)$ and $C_2(\delta)$ depending only on $\delta$, such that
\[
\pp(\textup{The recovered polytree is a valid DAG for $X$}) \ge 1-  C_1(\delta)p^3 e^{-C_2(\delta) n}. 
\]
\end{thm}
The first assumption above is a quantitative version of the common assumption that if $j\to i\leftarrow k$, then $X_j$ and $X_k$ are conditionally dependent given $X_i$, and if $j\to i\to k$ or $j\leftarrow i\to k$ or $j\leftarrow i\leftarrow k$, then $X_j$ and $X_k$ are unconditionally dependent. The second assumption is a quantitative version of the assumption that if $i$ and $j$ are neighbors, then $X_i$ and $X_j$ are dependent (essentially, the same as the first assumption of Theorem \ref{mainresult}). The third assumption is reasonable for the same reasons as discussed for the analogous assumption in Theorem \ref{mainresult}.

\section{Examples}
\subsection{Simulations}\label{simul}
In this section, we present simulation results for four kinds of polytrees. The code used for all of the following is available at \url{https://souravchatterjee.su.domains/condep.R}. The first is the linear tree, which is just a sequence of nodes arranged in a line, with adjacent ones joined by edges. If there are $p$ nodes, the random variables $X_1,\ldots,X_p$ are generated according to the recursion
\[
X_i = \frac{X_{i-1} + \ep_i}{\sqrt{2}}
\]
with $X_1 = \ep_1$, where $\ep_1,\ldots,\ep_p$ are i.i.d.~$N(0,1)$ random variables.

The second is the binary tree, which has a root node with has two children, each child has two children of its own, and so on. If $p$ is the number of nodes, we denote them by $1,\ldots, p$, with $1$ denoting the root node. If $i$ is the parent of a node $j$, we define
\[
X_j = \frac{X_i + \ep_j}{\sqrt{2}},
\]
with $X_1 = \ep_1$, where again $\ep_1,\ldots,\ep_p$ are i.i.d.~$N(0,1)$ random variables.

The third is the star tree, which has one central node and the other nodes are all children of the central node. If the nodes are marked $1,\ldots, p$, we let $1$ be the central node. We let $X_1 = \ep_1$, and for each $i\ge 2$, we let
\[
X_i = \frac{X_1 + \ep_i}{\sqrt{2}},
\]
where $\ep_1,\ldots,\ep_p$ are i.i.d.~$N(0,1)$ random variables. 

Finally, the fourth is the reverse binary tree, which is exactly the same as the binary tree, but the random variables are defined differently. Here, we let $X_i$ to be an independent $N(0,1)$ random variable for each leaf $i$. Then, we define the $X_i$'s for the remaining nodes by backward induction as follows. If $i$ is an internal node with children $j$ and $k$, we define
\[
X_i = \frac{X_j + X_k + \ep_i}{\sqrt{3}},
\]
where $\ep_i$'s are i.i.d.~$N(0,1)$ random variables. It is not hard to see that the trees described in each of the above examples are indeed causal polytree skeletons of the $X_i$'s.

\begin{table}[t!]
\centering
\caption{Proportion of undirected skeleton edges correctly identified by our algorithm. \label{maintable}}
\begin{footnotesize}
\begin{tabular}{rrrrrrr}
  \toprule
  & & & \multicolumn{4}{c}{Sample size $n$} \\
  \cmidrule{4-7}
Tree type & Tree size $p$ & & $50$ & $100$ & $200$ & $300$ \\ 
\midrule
\ldelim\{{3}{*}[Linear] & $15$ & & $0.82$ & $0.96$ & $0.99$ & $1.00$\\
& $511$ & & $0.82$ & $0.94$ & $1.00$ & $1.00$ \\
& $1023$ & & $0.70$ & $0.93$ & $1.00$ & $1.00$ \\
\ldelim\{{3}{*}[Binary] & $15$ & & $0.80$ & $0.93$ & $0.99$ & $1.00$\\
& $511$ &  & $0.69$ & $0.92$ & $0.99$ & $1.00$ \\
& $1023$ & & $0.67$ & $0.93$ & $0.99$ & $1.00$ \\ 
\ldelim\{{3}{*}[Star] & $15$ & & $0.55$ & $0.83$ & $0.98$ & $1.00$\\
& $511$ & & $0.09$ & $0.33$ & $0.77$ & $0.94$ \\
& $1023$ & & $0.05$ & $0.24$ & $0.70$ & $0.92$ \\
\ldelim\{{3}{*}[Reverse binary] & $15$ & & $0.68$ & $0.81$ & $0.97$ & $0.99$ \\
& $511$ & & $0.39$ & $0.69$ & $0.94$ & $0.99$ \\
& $1023$ & & $0.32$ & $0.65$ & $0.92$ & $0.99$ \\
   \bottomrule
\end{tabular}
\end{footnotesize}
\end{table}

\begin{table}[t!]
\centering
\caption{Proportion of directed edges correctly identified by our algorithm. \label{maintable2}}
\begin{footnotesize}
\begin{tabular}{rrrrrrr}
  \toprule
  & & & \multicolumn{4}{c}{Sample size $n$} \\
  \cmidrule{4-7}
Tree type & Tree size $p$ & & $50$ & $100$ & $200$ & $300$ \\ 
\midrule
\ldelim\{{3}{*}[Linear] & $15$ & & $0.54$ & $0.80$ & $0.94$ & $0.96$\\
& $511$ & & $0.37$ & $0.49$ & $0.78$ & $0.95$ \\
& $1023$ & & $0.36$ & $0.47$ & $0.66$ & $0.91$ \\
\ldelim\{{3}{*}[Binary] & $15$ & & $0.56$ & $0.68$ & $0.83$ & $0.86$\\
& $511$ &  & $0.42$ & $0.66$ & $0.81$ & $0.88$ \\
& $1023$ & & $0.41$ & $0.65$ & $0.81$ & $0.88$ \\ 
\ldelim\{{3}{*}[Star] & $15$ & & $0.45$ & $0.64$ & $0.81$ & $0.85$\\
& $511$ & & $0.08$ & $0.28$ & $0.72$ & $0.90$ \\
& $1023$ & & $0.05$ & $0.22$ & $0.66$ & $0.89$ \\
\ldelim\{{3}{*}[Reverse binary] & $15$ & & $0.42$ & $0.59$ & $0.84$ & $0.88$ \\
& $511$ & & $0.21$ & $0.44$ & $0.73$ & $0.85$ \\
& $1023$ & & $0.17$ & $0.39$ & $0.72$ & $0.85$ \\
   \bottomrule
\end{tabular}
\end{footnotesize}
\end{table}

For each of these models, we try out our algorithm with various values of $n$ and $p$. The accuracy of the output is measured by the proportion of edges that are correctly identified. Since any spanning tree has $p-1$ edges, this is a reasonable measure of discrepancy. The average value of this proportion in $20$ simulations is then calculated for each case. The results for skeleton recovery are tabulated in Table \ref{maintable}. Since $p$ must be of the form $2^k -1$ for the binary tree example, we take all our $p$'s of this form. To be specific, we take $p=15$, $p=511$ and $p=1023$, and for each $p$, we consider $n= 50$, $n=100$, $n=200$ and $n=300$.

From Table \ref{maintable}, we see that for the linear tree and the binary tree, our skeleton recovery algorithm performs very well (more than $90\%$ edges correctly recovered) even for $n=100$ and $p = 1023$. The star tree and the reverse binary tree are harder to estimate, with the reverse binary tree requiring $n=200$ for a recovery rate greater than $90\%$ (for $p=1023$), and the star tree requiring $n=300$.

Table \ref{maintable2} shows the proportions of directed edges that are correctly identified by our algorithm for recovering directionalities. This is less accurate than the skeleton recovery, but still quite satisfactory. For example, with $n=300$, the rate of correct recovery  is greater than $85\%$ in all four types of trees even when $p=1023$.

Both tables show that the power of our algorithms depend mainly on the sample size, rather than the number of variables. Indeed, if one goes down the column for $n=300$, there seems to be no significant deterioration in performance as $p$ increases from $15$ to $1023$. This is in agreement with our theoretical error bounds, which are polynomially increasing in $p$ and exponentially decreasing in $n$. We conjecture that for $n$ of the order of $500$ or so, the algorithms will perform well even with extremely large $p$ (such as $20000$). Verifying this conjecture will be computationally rather expensive, which is what prevented us from doing it for this paper.

\subsection{Real data}
We now consider a real data example, from the R data repository `causaldata'~\cite{causaldata}, on the effect of mortgage subsidies on home ownership. The `mortgages' data  in the causaldata package contains data from Fetter~\cite{fetter13} on home ownership rates by men, focusing on whether they were born at the right time to be eligible for mortgage subsidies based on their military service. The data consists of 214,144 observations of the following 6 variables:
\begin{enumerate}
\item Birth state.
\item Quarter of birth
\item Race (white/nonwhite).
\item Whether veteran of either the Korean war or World War II.
\item Owns a home or not.
\item Quarter of birth centered on eligibility for mortgage subsidy (0+ = eligible).
\end{enumerate}
All variables other than birth state are numerical. Since our algorithm can only handle numerical variables, birth state was converted to a number ranging from 1 to 50. The directed tree produced by our algorithm is shown in Figure \ref{mortgagesfig}. The tree shows a causal structure that is acceptable to common sense: Birth state has a causal effect on race, which, in turn, has a causal effect on  home ownership. Veteran status and quarter of birth affect mortgage subsidy eligibility, which affects home ownership. 
\begin{figure}
\begin{center}
\includegraphics[scale=.5]{"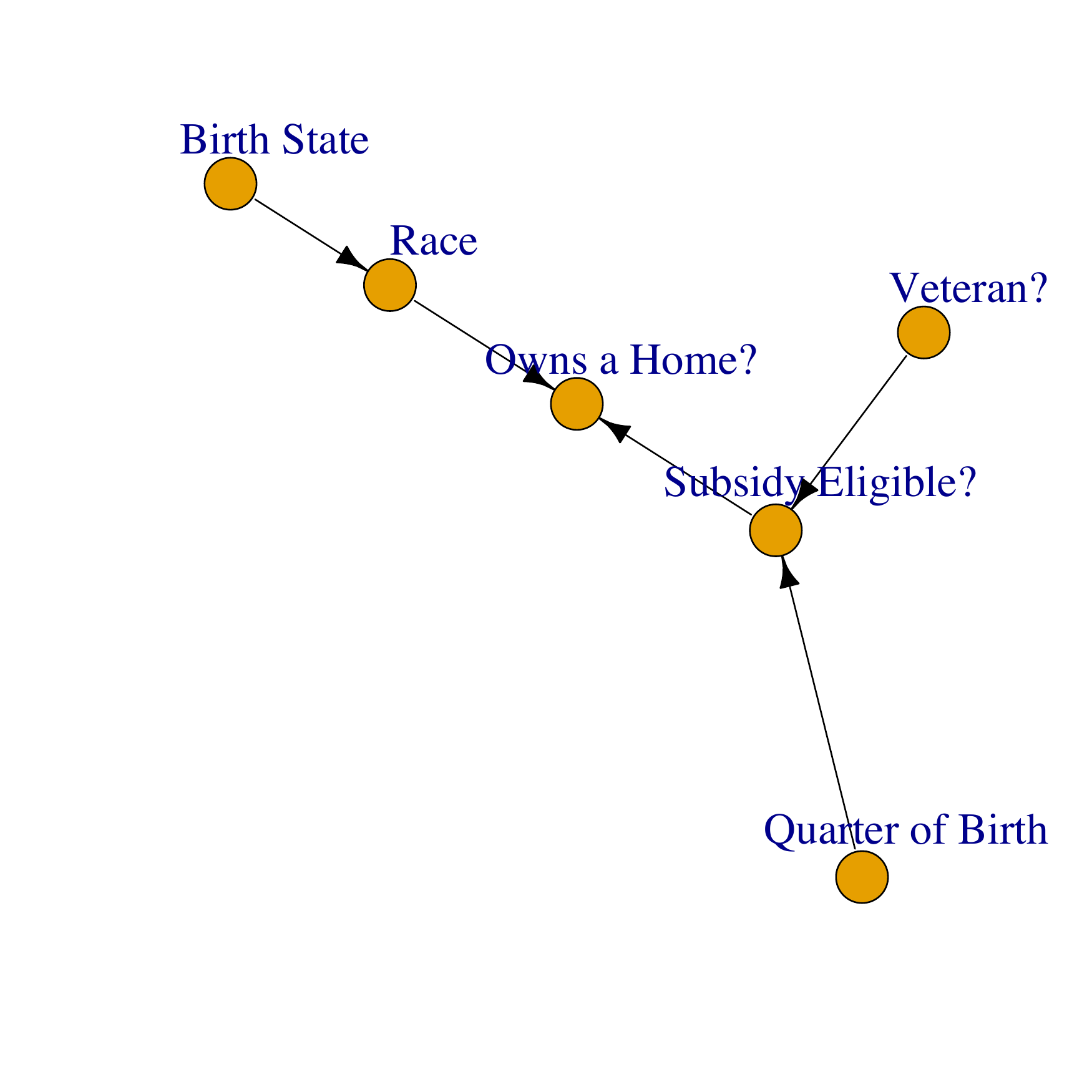"}
\caption{Estimated causal polytree for the mortgage data.\label{mortgagesfig}}
\end{center}
\end{figure}

\section{Proof of Theorem \ref{mainresult}}\label{proofsec}
The proof of Theorem \ref{mainresult} requires several steps. These are divided into subsections below, for the reader's convenience.
\subsection{A property of causal polytrees}\label{causal}
The following important property of causal polytrees will be used in many places. It is a well-known fact, but we give a proof for the sake of completeness.
\begin{prop}\label{causalprop}
Let $T$ be the skeleton of a causal polytree skeleton for a collection $X = (X_i)_{i\in V}$. Let $i,j,k$ be three vertices such that $j$ lies on the path connecting $i$ and $k$ in  $T$. Then at least one of the following is true: (a) $X_i$ and $X_k$ are conditionally independent given $X_j$, or (b) $X_i$ and $X_k$ are unconditionally independent.
\end{prop}
\begin{proof}
Let the function $f$ displayed in equation \eqref{density} be the joint probability density of $X$ with respect to some product measure. Take any distinct $i,j,k\in V$ such that $j$ lies on the path connecting $i$ and $k$ in $T$. 

For each $a\in V\setminus\{j\}$, let $q(a)$ be the unique neighbor of $j$ which belongs to the path connecting $a$ and $j$ in $T$. Let $b_1,\ldots, b_m$ be the parents of $j$ (where $m$ may be $0$), and for each $1\le l\le m$, let $V_l$ be the set of all $a$ such that $q(a) = b_l$. Let $V_0$ be the set of all $a$ such that $q(a)$ is not a parent of $j$. Let $U :=  V_1\cup \cdots \cup V_m$. 

Now note that in the product displayed in \eqref{density}, $x_a$ and $x_b$ can be in the same factor only if either $a$ is a parent of $b$, or $b$ is a parent of $a$, and $a$ and $b$ have a common child. Thus, if $a\in V_0$ and $b\in U$, then $x_a$ and $x_b$ cannot appear in a common factor. This shows that the collections $(X_a)_{a\in V_0}$ and $(X_b)_{b\in U}$ are conditionally independent given $X_j$. In particular, if $i\in V_0$ and $k\in U$, or if $i\in U$ and $k\in V_0$, then $X_i$ and $X_k$ are conditionally independent given $X_j$. 

The only other possibility is that $i\in V_l$ and $k\in V_{l'}$ for some $1\le l\ne l'\le m$. We will show that in this case, $X_i$ and $X_k$ are unconditionally independent. First, take any $a\in V_0$ which has no children. Then we can integrate out $x_a$ in the product \eqref{density} and remove it from the graph. The resulting model is again described by a DAG whose skeleton is a forest, because we are simply removing a vertex and the edges incident to it. 

Continuing like this, we can iteratively remove all vertices in $V_0$, leaving us with only $U\cup \{j\}$. Integrating out $x_j$, we now see that the collections $(X_a)_{a\in V_l}$, $l=1,\ldots, m$  are independent. This shows that $X_i$ and $X_k$ are independent. 
\end{proof}

\subsection{Concentration of $\xi_n$}
Let $(X,Y)$ be a pair of real-valued random variables. Let $n\ge 2$ and let $(X_i,Y_i)$, $i=1,\ldots, n$ be i.i.d.~copies of $(X,Y)$. Let $\xi_n$ be the $\xi$-coefficient between $X_1,\ldots,X_n$ and $Y_1,\ldots,Y_n$. Recall the $\alpha$-coefficient defined in Section \ref{result}. The goal of this subsection is to prove the following result. 
\begin{prop}\label{concprop}
Suppose that $\alpha(Y)$ is bounded below by some $\delta>0$. Then, there are positive numbers $C_1(\delta)$, $C_2(\delta)$ and $C_3(\delta)$ depending only on $\delta$ such that for all $t\in [0, C_1(\delta)]$ and all $n$,
\begin{align*}
\pp(|\xi_n - \ee(\xi_n)|\ge t) \le C_2(\delta) e^{-C_3(\delta) nt^2}. 
\end{align*}
\end{prop}
This immediately yields the following corollary.
\begin{cor}\label{xicor}
If $Y$ is not a constant, then as $n\to \infty$, $\ee(\xi_n)\to \xi$, where $\xi$ is the $\xi$-correlation between $X$ and $Y$ displayed in equation \eqref{xicordef}.
\end{cor}
\begin{proof}
By \cite[Theorem 1.1]{chatterjee21}, $\xi_n \to \xi$ in probability as $n\to \infty$. By Proposition~\ref{concprop}, $\xi_n - \ee(\xi_n) \to 0$ in probability. The claim is proved by combining these observations. 
\end{proof}
To prove Proposition \ref{concprop}, we need two lemmas. The first one is the following. 
\begin{lmm}\label{betalmm}
Let $\alpha := \alpha(Y)$ and let $\mu$ be the law of $Y$. For each $t\in \rr$, let $G(t) := \pp(Y\ge t)$. Let
\begin{align*}
\beta := \int G(t)(1-G(t)) d\mu(t). 
\end{align*}
Then $\beta \ge \alpha^2/32$.
\end{lmm}
\begin{proof}
Let $Y'$ be an independent copy of $Y$. Then note that 
\[
\pp(Y>Y') = \pp(Y<Y'),
\] 
and thus,
\[
\pp(Y'> Y) = \frac{1}{2}(\pp(Y' >Y) + \pp(Y' < Y)) = \frac{1}{2}(1-\pp(Y'=Y)).
\]
This gives
\[
\pp(Y'\ge Y) = \pp(Y'>Y) + \pp(Y'=Y) = \frac{1}{2}(1+\pp(Y'=Y)) = 1 - \frac{\alpha}{2}.
\]
But  $\pp(Y'\ge Y) = \ee(G(Y))$. Thus, 
\begin{align}\label{gmu}
\int G(t) d\mu(t) = 1-\frac{\alpha}{2}. 
\end{align} 
Take any $\ep\in (0,1/2)$. Define three sets 
\[
A_\ep := \{t: G(t) > 1-\ep\}, \ B_\ep := \{t: \ep \le G(t) \le  1-\ep\}, \  C_\ep := \{t: G(t) < \ep\}.
\]
Then 
\begin{align*}
\int G(t) d\mu(t) \ge (1-\ep) \mu(A_\ep).
\end{align*}
Combining this with \eqref{gmu},  we get
\begin{align*}
\mu(A_\ep) \le \frac{1}{1-\ep}\biggl(1-\frac{\alpha}{2}\biggr). 
\end{align*}
On the other hand, since $G$ is a non-increasing left-continuous function, the set $C_\ep$ is an interval of the form $(t,\infty)$. Thus,
\begin{align*}
\mu(C_\ep) &= \mu((t,\infty)) = \lim_{s\downarrow t} G(s) \le \ep.
\end{align*}
Combining the last two displays, we get
\begin{align*}
\mu(B_\ep) = 1- \mu(A_\ep)-\mu(B_\ep) \ge  1 - \frac{1}{1-\ep}\biggl(1-\frac{\alpha}{2}\biggr) - \ep.
\end{align*}
But $G(t)(1-G(t))\ge \ep(1-\ep)$ on $B_\ep$ since $\ep\in (0,1/2)$. Thus,
\begin{align*}
\beta &\ge \ep(1-\ep)\mu(B_\ep) \\
&\ge \ep(1-\ep) - \ep\biggl(1-\frac{\alpha}{2}\biggr) - \ep^2(1-\ep)\\
&= \frac{\ep\alpha}{2} -2\ep^2 + \ep^3\ge \frac{\ep\alpha}{2}-2\ep^2.
\end{align*}
Choosing $\ep = \alpha/8$ completes the proof.
\end{proof}
The second lemma we need is the following simple upper bound on $\xi_n$.
\begin{lmm}\label{xibound}
For any $n$ and any realization of the data, $|\xi_n|\le 1+n^2$.
\end{lmm}
\begin{proof}
If $Y_1=Y_2=\cdots=Y_n$ in a particular realization of the data, then $\xi_n =0$ by the convention introduced in Section \ref{algo}. So, let us assume that not all the $Y_i$'s are equal. Recall the $r_i$'s and $l_i$'s used in the definition of $\xi_n$ in  Section \ref{algo}. Since the $Y_i$'s are not all equal, it is easy to see that there is at least one $l_i$ which lies strictly between $0$ and $n$. Thus, for this $i$, $l_i(n-l_i)\ge n$. Since $|r_{i+1}-r_i|\le n$ for each $i$, inspecting the formula \eqref{xidef} now shows that $|\xi_n|\le 1+n^2$.
\end{proof}
We are now ready to prove Proposition \ref{concprop}.
\begin{proof}[Proof of Proposition \ref{concprop}]
Define
\[
G_n(t) := \frac{1}{n}\sum_{i=1}^n 1_{\{Y_i\ge t\}}. 
\]
Define 
\begin{align*}
S_n := \frac{1}{n}\sum_{i=1}^n G_n(Y_i)(1-G_n(Y_i)), \ \ \ S_n' := \frac{1}{n}\sum_{i=1}^n G(Y_i)(1-G(Y_i)). 
\end{align*}
Recall the number $\beta$ defined in Lemma \ref{betalmm}. Note that $\ee(S_n')=\beta$ and $S_n'$ is an average of i.i.d.~random variables taking values in $[0,1]$. Therefore, by Hoeffding's inequality~\cite{hoeffding63},
\begin{align*}
\pp(|S_n' - \beta|\ge t) \le 2e^{-2nt^2}
\end{align*}
for any $t\ge 0$. Let $\Delta_n := \sup_{t\in \rr} |G_n(t)-G(t)|$. By the Dvoretzky--Kiefer--Wolfowitz inequality~\cite{dkw56, massart90},
\[
\pp(\Delta_n \ge t) \le 2e^{-2nt^2}
\]
for any $t\ge 0$. Now, by the triangle inequality, $|S_n-S_n'|\le 2\Delta_n$. Combining all of this, we get that for any $t\ge 0$, 
\begin{align*}
\pp(|S_n -\beta|\ge t) &\le \pp( |S_n - S_n'|\ge t/2) + \pp(|S_n'-\beta|\ge t/2)\\
&\le \pp(\Delta_n \ge t/4) + \pp(|S_n'-\beta|\ge t/2)\le 4e^{-nt^2/8}. 
\end{align*}
From the proof of \cite[Theorem 1.1]{chatterjee21} in the supplementary materials of \cite{chatterjee21}, recall that 
\begin{align}\label{qnsn}
\biggl|\frac{Q_n}{S_n} - \xi_n\biggr| &\le \frac{1}{2nS_n},
\end{align}
where $Q_n$ is a random variable taking values in $[-1,1]$, which has a somewhat complicated definition that is unnecessary to state here. The only thing we need to know about $Q_n$ is that for all $t\ge 0$,
\[
\pp(|Q_n - \ee(Q_n)|\ge t) \le 2e^{-Cnt^2},
\]
where $C$ is a positive universal constant. We get this from Lemma A.11 in the supplementary materials of \cite{chatterjee21}. Now, if $|Q_n-\ee(Q_n)|< t$ and $|S_n - \beta|< t$ for some $t\in [0,\beta/2]$, then $S_n> \beta/2$, and hence 
\begin{align*}
\biggl|\frac{Q_n}{S_n} - \frac{\ee(Q_n)}{\beta}\biggr| &= \frac{|Q_n \beta - S_n \ee(Q_n)|}{S_n \beta}\\
&\le \frac{2}{\beta^2}(|Q_n-\ee(Q_n)|\beta + \ee(Q_n)|S_n-\beta|)\\
&< \frac{2(\beta+1)t}{\beta^2},
\end{align*}
and consequently, by \eqref{qnsn}, 
\begin{align*}
\biggl|\xi_n - \frac{\ee(Q_n)}{\beta}\biggr| &< \frac{2(\beta+1)t}{\beta^2} + \frac{1}{2nS_n}\\
&\le \frac{2(\beta+1)t}{\beta^2} + \frac{1}{n\beta}.
\end{align*}
Let $a_n := \ee(Q_n)/\beta$. From the above, we see that for all $t\in [0,\beta/2]$,
\begin{align}\label{xian}
\pp\biggl(|\xi_n - a_n|\ge  \frac{2(\beta+1)t}{\beta^2} + \frac{1}{n\beta}\biggr) &\le 6e^{-Cnt^2}
\end{align}
where $C$ is a positive universal constant. By Lemma \ref{betalmm} and Lemma \ref{xibound}, it is easy to see using the above inequality that 
\begin{align*}
|\ee(\xi_n)-a_n|\le \ee|\xi_n-a_n| = \int_0^\infty\pp(|\xi_n-a_n|\ge t) dt &\le \frac{C(\delta)}{\sqrt{n}},
\end{align*}
where $C(\delta)$ is a positive constant that depends only on $\delta$. 
Thus, again using \eqref{xian}, we get that there are positive numbers $C_1(\delta)$, $C_2(\delta)$ and $C_3(\delta)$ depending only on $\delta$ such that for all $t\in [0, C_1(\delta)]$ and all $n$,
\begin{align*}
\pp\biggl(|\xi_n - \ee(\xi_n)|\ge t + \frac{C_2(\delta)}{\sqrt{n}} \biggr) \le 6 e^{-C_3(\delta) nt^2}. 
\end{align*}
Now, if $t\ge 2C_2(\delta)/\sqrt{n}$ and $t\le C_1(\delta)$, the above inequality implies that
\begin{align*}
\pp(|\xi_n - \ee(\xi_n)|\ge t) &\le \pp\biggl(|\xi_n - \ee(\xi_n)|\ge \frac{t}{2} + \frac{C_2(\delta)}{\sqrt{n}} \biggr) \\
&\le 6 e^{-C_3(\delta) nt^2/4}. 
\end{align*}
On the other hand, if $t< 2C_2(\delta)/\sqrt{n}$, then 
\begin{align*}
\pp(|\xi_n - \ee(\xi_n)|\ge t) &\le 1 \le e^{4C_3(\delta)C_2(\delta)^2} e^{-C_3(\delta)nt^2}. 
\end{align*}
Combining the last two displays completes the proof.
\end{proof}
\subsection{Data processing inequality for maximal correlation}
In this subsection we prove that the maximal correlation coefficient satisfies an inequality that is also satisfied by mutual information, going by the name `data processing inequality'.
\begin{prop}\label{dpimaxcor}
Let $X$, $Y$ and $Z$ be three random variables such that $Y$ and $Z$ are conditionally independent given $X$. Then $R(Z,Y)\le R(X,Y)$. 
\end{prop}
\begin{proof}
If $Y$ is a constant or $Z$ is a constant, then $R(Z,Y)=0$ by definition, and so there is nothing to prove. So, let us assume that both $Y$ and $Z$ are non-constant. In this case, if $X$ is a constant, then the conditional independence of $Y$ and $Z$ given $X$ implies that $Y$ and $Z$ are unconditionally independent, and hence $R(Z,Y)=0$. So, again, there is nothing to prove. Thus, let us assume that $X$ is also non-constant. 

Take any $f$ and $g$ such that $\var(f(Z))$ and $\var(g(Y))$ are both in $(0,\infty)$. Without loss of generality, $\ee(f(Z))=\ee(g(Y))=0$ and $\var(f(Z)) = \var(g(Y)) = 1$. Then by the conditional independence of $Y$ and $Z$ given $X$, we have
\begin{align*}
\mathrm{Corr}(f(Z), g(Y)) = \ee(f(Z)g(Y)) = \ee(h(X)g(Y)),
\end{align*}
where $h(X) := \ee(f(Z)|X)$. But $\ee(h(X)) = \ee(f(Z)) = 0$, and so
\[
\ee(h(X) g(Y)) \le R(X,Y) \sqrt{\ee(h(X)^2) \ee(g(Y)^2)} = R(X,Y) \sqrt{\ee(h(X)^2)}. 
\]
To complete the proof, note that $\ee(h(X)^2)\le \ee(f(Z)^2) = 1$. 
\end{proof}
\subsection{Data processing inequality for $\xi$-correlation}
In this section we will show that the $\xi$-correlation also satisfies a data processing inequality, although it is a `one-sided' version of the inequality, since the $\xi$-correlation is not symmetric. More importantly, the result makes a connection between the $\xi$-correlation and the maximal correlation, which will be important in the proof of Theorem \ref{mainresult}.
\begin{prop}\label{xidpi}
Let $X$, $Y$ and $Z$ be three random variables such that $Y$ and $Z$ are conditionally independent given $X$. Then $\xi(Z,Y)\le R(Z,X)^2\xi(X,Y)$. In particular, $\xi(Z,Y)\le \xi(X,Y)$. 
\end{prop}
\begin{proof}
If $Y$ is a constant, then both $\xi(Z,Y)$ and $\xi(X,Y)$ are zero (according to the convention adopted in Section \ref{result}), so there is nothing to prove. Thus, let us assume that $Y$ is not a constant. 
Take any $t\in \rr$. Let $f(X) := \pp(Y\ge t|X)$ and $g(Z) := \pp(Y\ge t|Z)$. By the conditional independence of $Y$ and $Z$ given $X$, $g(Z) = \ee(f(X)|Z)$. Let $\tf(X) := f(X)-\ee(f(X))$ and $\tg(Z) := g(Z)-\ee(g(Z))$. Since $\ee(f(X))=\ee(g(Z))$, we have $\ee(\tf(X)|Z) = \tg(Z)$. Thus, 
\begin{align*}
\var(g(Z)) &= \ee(\tg(Z)^2) \\
&= \ee(\tg(Z) \ee(\tf(X)|Z)) \\
&= \ee(\tg(Z)\tf(X)) \le R(Z,X) \sqrt{\ee(\tg(Z)^2) \ee(\tf(X)^2)}.
\end{align*}
Rearranging this inequality, we get
\[
\var(g(Z)) \le R(Z,X)^2 \var(f(X)).
\]
Since $t$ is arbitrary, this implies that
\begin{align*}
\xi(Z,Y) &= \frac{\int \var(\pp(Y\ge t|Z)) d\mu(t)}{\int \var(1_{\{Y\ge t\}}) d\mu(t) } \\
&\le \frac{\int R(Z,X)^2\var(\pp(Y\ge t|X)) d\mu(t)}{\int \var(1_{\{Y\ge t\}}) d\mu(t) } = R(Z,X)^2 \xi(X,Y),
\end{align*}
which completes the proof. 
\end{proof}
\subsection{A fact about maximal spanning trees}
The final bit of preliminary factoid that we need is the following result about maximal weighted spanning trees.
\begin{lmm}\label{mstlmm}
Let $T$ be a maximal weighted spanning tree of a connected weighted graph $G = (V,E)$. If an edge $e$ in $G$ does not belong to $T$, then every edge in the path connecting the endpoints of $e$ in $T$ must have weight greater than or equal to the weight of $e$.
\end{lmm}
\begin{proof}
Suppose that some edge $f$ in the path connecting the endpoints of $e$ in $T$ has $w_f < w_e$. Then deleting $f$ and adding $e$ to $T$ gives us a spanning tree whose weight is strictly greater than that of $T$, which contradicts the maximality of $T$.
\end{proof}
This completes the preliminary steps. We are now ready to prove Theorem \ref{mainresult}.
\subsection{Proof of Theorem \ref{mainresult}}
First, take any distinct $i,j\in V$. Let $k$ be the vertex adjacent to $j$ on the path joining $i$ to $j$ in $T$. Then by Proposition \ref{causalprop}, Proposition \ref{dpimaxcor} and the assumption that $R_{kj}\le 1-\delta$ from the statement of Theorem~\ref{mainresult}, we get that $R_{ij}\le R_{kj} \le 1-\delta$. Thus, 
\begin{align}\label{rijbound}
R_{ij} \le 1-\delta \ \ \text{for all distinct $i,j\in V$.}
\end{align}
Take any distinct $i,j,k\in V$ such that $j$ is on the path connecting $i$ and $k$ in $T$. Then by Proposition \ref{causalprop}, either $X_i$ and $X_k$ are conditionally independent given $X_j$, or $X_i$ and $X_k$ are unconditionally independent. In the first case, by Proposition \ref{xidpi} and the inequality \eqref{rijbound}, 
\begin{align}\label{xiki}
\xi_{ki}\le R_{kj}^2\xi_{ji}\le (1-\delta)^2 \xi_{ji}.
\end{align}
In the second case, $\xi_{ki}=0$, and so the above inequality holds anyway. If $i$ and $j$ are neighbors, this shows that
\begin{align}\label{xigamma}
\xi_{ki} \le \xi_{ji} - (2\delta - \delta^2) \xi_{ji} \le \xi_{ji} - (2\delta - \delta^2)\delta,
\end{align}
because $2\delta - \delta^2 > 0$ and $\xi_{ji}\ge\delta$. Define $\gamma := (2\delta - \delta^2)\delta$.
Note that since $\delta >0$, $2\delta > \delta^2$ and $2\delta - \delta^2 \le 1$ (because $(1-\delta)^2 \ge 0$), we have
\begin{align}\label{gammaineq}
0 < \gamma \le \delta. 
\end{align}
We claim that if in a particular realization of $X^1,\ldots,X^n$, we have 
\begin{align}\label{xincond}
|\xi_{ij}^n - \xi_{ij}|\le \frac{\gamma}{4} \ \ \text{ for all $i,j$,}
\end{align} 
then $T_n=T$. To see this, suppose that \eqref{xincond} holds in some realization of the data. First, taking any edge $(i,j)$ in $T$, let us show that $(i,j)$ is in $G_n$. Deleting the edge $(i,j)$ from $T$ splits $T$ into two disjoint connected components $A$ and $B$, with $i\in A$ and $j\in B$. Take any $k\in A\setminus \{i\}$. Then by the definition of $\gamma$ and the inequality~\eqref{xigamma}, we have $\xi_{ij} \ge \xi_{kj} +\gamma$. Thus, by \eqref{xincond} and~\eqref{gammaineq},
\begin{align}
\xi_{kj}^n &\le \xi_{kj} + \frac{\gamma}{4}\le \xi_{ij} -\gamma +  \frac{\gamma}{4}\notag \\
&\le \xi_{ij}^n+\frac{\gamma}{4}-\gamma + \frac{\gamma}{4} < \xi_{ij}^n.\label{correct}
\end{align}
Similarly, if $k\in B\setminus \{j\}$, $\xi_{ki}^n < \xi_{ji}^n$. Therefore, for any $k\in V\setminus\{i,j\}$, the condition~\eqref{mainxicond} is violated. This shows that $(i,j)$ is an edge in $G_n$.

Thus, $T$ is a subgraph of $G_n$ under \eqref{xincond}. This shows, in particular, that $G_n$ is connected. Now recall that $T_n$ is an MWSF of $G_n$ when the edges are endowed with the weights $w_{ij}^n$. By  \eqref{gammaineq} and the definition of $\delta$, we have
\begin{align*}
w_{ij}\ge \delta\ge \gamma>0  \ \ \textup{for all edges $(i,j)$ in $T$.}
\end{align*} 
Thus, under~\eqref{xincond}, by \eqref{gammaineq} we get
\begin{align}\label{wnineq}
w_{ij}^n \ge \delta - \frac{\gamma}{4} > \frac{\delta}{2} > 0 \ \ \textup{for all edges $(i,j)$ in $T$.}
\end{align}
Since $T_n$ is an MWSF of the connected graph $G_n$ and every edge-weight is strictly positive (by \eqref{wnineq}), we deduce that $T_n$ is a maximal weighted spanning tree of $G_n$ under \eqref{xincond}. Moreover, as noted above, $T$ is a spanning tree of $G_n$. Thus, to show that $T=T_n$, we only need to prove that any edge of $G_n$ that is not in $T$ cannot be in $T_n$.

So, take any edge $(i,j)$ of $G_n$ that is not in $T$ but is in $T_n$. We will prove by contradiction that such an edge cannot exist. First, we claim that $w_{ij}^n > \delta/2$. Let $P$ be the path in $T$ that connects $i$ to $j$. Then by \eqref{wnineq}, we have that $w_{kl}^n > \delta/2$ for every edge $(k,l)$ in $P$.

Since $(i,j)$ is an edge in $T_n$, deleting the edge $(i,j)$ splits $T_n$ into two components $A$ and $B$, with $i\in A$ and $j\in B$. Since $P$ connects $i$ to $j$ in $T$, there must be at least one edge $(k,l)$ in $P$ such that $k\in A$ and $l\in B$. Then note that $(k,l)$ is not an edge of $T_n$, and the path joining $k$ to $l$ in $T_n$ contains $(i,j)$. By Lemma \ref{mstlmm} and the fact that $(i,j)$ is in $G_n$, this proves that $w_{ij}^n \ge w_{kl}^n > \delta/2$. 

Next, to get a contradiction, we will show that $w_{ij}^n \le  \delta/2$. Take any vertex $k$ in $P$ that is not $i$ or $j$. Since $(i,j)$ is an edge of $G_n$, it follows that either $\xi_{ki}^n < \xi_{ji}^n$ or $\xi_{kj}^n < \xi_{ij}^n$. Suppose that $\xi_{ki}^n < \xi_{ji}^n$. Then by \eqref{xincond} and \eqref{xiki},
\begin{align*}
\xi_{ji}^n &\le \xi_{ji} + \frac{\gamma}{4} \le (1-\delta)^2 \xi_{ki} + \frac{\gamma}{4}\\
&\le (1-\delta)^2 \biggl(\xi_{ki}^n + \frac{\gamma}{4}\biggr)+ \frac{\gamma}{4}\\
&\le (1-\delta)^2 \biggl(\xi_{ji}^n + \frac{\gamma}{4}\biggr)+ \frac{\gamma}{4}.
\end{align*}
Rearranging this inequality,  we get
\begin{align*}
\xi_{ji}^n &\le \frac{1}{2\delta - \delta^2} \biggl((1-\delta)^2\frac{\gamma}{4}+\frac{\gamma}{4}\biggr)\\
&\le \frac{\gamma}{2(2\delta - \delta^2)} = \frac{\delta}{2}. 
\end{align*}
Similarly, if $\xi_{kj}^n <\xi_{ij}^n$, we get $\xi_{ij}^n \le \delta/2$. Combining the two cases, we have that $w_{ij}^n \le \delta/2$, which gives the desired contradiction that proves that $T_n=T$ if \eqref{xincond} holds in some realization of the data.

Now note that by assumption, we have that for all $i,j\in V$,
\begin{align*}
|\ee(\xi_{ij}^n) - \xi_{ij}| &\le \frac{\delta^2}{8} \le \frac{\delta^2(2-\delta)}{8} = \frac{\gamma}{8}.
\end{align*}
Thus, for \eqref{xincond} to hold, it suffices that we have
\begin{align*}
|\xi_{ij}^n - \ee(\xi_{ij}^n)|\le \frac{\gamma}{8} \ \ \text{ for all $i,j$.}
\end{align*} 
To prove the required bound, it is therefore sufficient to show that for any distinct $i,j\in V$, 
\begin{align}\label{mainineq}
\pp\biggl(|\xi_{ij}^n - \ee(\xi_{ij}^n)|> \frac{\gamma}{8}\biggr) &\le C_1(\delta)e^{-C_2(\delta)n},
\end{align}
where $C_1(\delta)$ and $C_2(\delta)$ are positive constants that depend only on $\delta$. By Proposition \ref{concprop}, this holds if we replace $\gamma/8$ on the left side by any $t\in [0,C_3(\delta)]$, where $C_3(\delta)$ is another positive constant that depends only on $\delta$, and insert $t^2$ inside the exponent on the right. Taking $t= \min\{\gamma/8, C_3(\delta)\}$, and observing the left side can only increase if we replace $\gamma/8$ by $t$, we get the desired result.

\section{Proof of Theorem \ref{directionthm}} 
Consider a modified algorithm, where instead of working with the estimated skeleton, we apply the algorithm to assign directionalities to the edges of the actual skeleton. This modified algorithm cannot be implemented in practice because the actual skeleton is unknown, but Theorem~\ref{mainresult} implies that the modified algorithm gives the same output as the original algorithm with probability at least $1-C_1(\delta)p^2 e^{-C_2(\delta) n}$. Thus, it suffices to prove Theorem \ref{directionthm} for the modified algorithm. 

Next, we do a second modification of the algorithm. In Step 1 of the algorithm, we replace $\tau_{kji}^n$ by $\tau_{kji}$ and $\xi_{jk}^n$ by $\xi_{jk}$ everywhere. Again, this cannot be implemented in practice, because $\tau_{kji}$ and $\xi_{jk}$ are unknown. In fact, since we have already replaced the estimated skeleton by the actual skeleton in the first modification, the algorithm produced by the second modification has no dependence on the data at all. Nevertheless, the following is true.
\begin{lmm}\label{approxlmm}
Consider the algorithm produced by the two modifications described above. The tree obtained by running this algorithm is equal to the tree produced by the original algorithm with probability at least $1-C_1(\delta)p^3 e^{-C_2(\delta) n}$.
\end{lmm}
\begin{proof}
Throughout this proof, we will say that an event occurs with `high probability' if it occurs with probability at least $1-C_1(\delta)p^3 e^{-C_2(\delta) n}$. 
Take any distinct $i,j,k\in V$ such that $i$ and $j$ are neighbors in $T$. 
By \cite[Lemma 11.9]{azadkiachatterjee21} and \cite[Lemma 13.3]{azadkiachatterjee21}, 
\[
\pp\biggl(|q_{kji}^n - \ee(q_{kji}^n)| > \frac{\delta^3}{8}\biggr) \le C_1 e^{-C_2 n \delta^6},
\]
and the same bound holds for $\pp(|s_{ji}^n - \ee(s_{ji}^n)|>\delta^3/8)$. If none of these bad events happen, then by the assumptions that $s_{ij}\ge \delta$, $|\ee(q_{kji}^n) - q_{kji}| \le \delta^3/8$ and $|\ee(s_{ji}^n) - s_{ji}|\le \delta^3/8$, we get
\begin{align*}
|\tau_{kji}^n - \tau_{kji}| &\le \frac{|q_{kji}^n - q_{kji}|}{s_{ji}^n} + \frac{|q_{kji}| |s_{ji}^n - s_{ji}|}{s_{ji}^ns_{ji}}\\
&\le \frac{\delta^3/4}{\delta + \delta^3/4} + \frac{\delta^3/4}{(\delta+\delta^3/4)\delta}\\
&\le \frac{\delta^2}{4} + \frac{\delta}{4}\le \frac{\delta}{2}.
\end{align*}
Thus, with high probability, $|\tau_{kji}^n - \tau_{kji}| \le \delta/2$ for all distinct $i,j,k\in V$ such that $i$ and $j$ are neighbors in $T$.

Now, by one of the assumptions in Theorem \ref{mainresult}, $|\ee(\xi_{ij}^n) - \xi_{ij}| \le \delta^2/8\le \delta/8$ for all distinct $i,j\in V$. By Proposition \ref{concprop}, $|\xi_{ij}^n - \E(\xi_{ij}^n)|\le \delta/8$ for all distinct $i,j\in V$ with high probability. Combining, we have that with high probability, $|\xi_{ij}^n - \xi_{ij}| \le \delta/4$ for all distinct $i,j\in V$.

Let $E$ be the event that $|\tau_{kji}^n - \tau_{kji}| \le \delta/2$ for all distinct $i,j,k\in V$ such that $i$ and $j$ are neighbors in $T$, and that $|\xi_{ij}^n - \xi_{ij}| \le \delta/4$ for all distinct $i,j\in V$, and that $T_n=T$. By the previous paragraphs and Theorem \ref{mainresult}, we know that $E$ is an event of high probability. Thus, to complete the proof of the lemma, it suffices to show that the modified algorithm gives the same output as the original algorithm if $E$ happens.

So, let us assume that $E$ happens. Then, we claim that for any distinct $i,j,k\in V$ such that $j$ and $k$ are both neighbors of $i$ in $T$, $\tau_{kji}^n \ge \xi_{jk}^n$ if and only if $\tau_{kji} \ge \xi_{jk}$. Clearly, this will imply that the output produced by the modified algorithm is the same as the output produced by the original algorithm.

To prove this, take any valid set of directionalities of the edges of $T$. There are four possibilities:
\begin{enumerate}
\item[Case 1.] $j \to i \leftarrow k$. In this case, $X_j$ and $X_k$ are independent, which implies that $\xi_{jk}=0$, and $\tau_{kji}\ge \delta$ by assumption. Thus, $\tau_{kji}\ge \xi_{jk}$, and by $E$, 
\[
\tau_{kji}^n \ge \tau_{kji} - \frac{\delta}{2}\ge \frac{\delta}{2}= \xi_{jk} + \frac{\delta}{2}\ge \xi_{jk}^n.
\]
\item[Case 2.] $j \to i \to k$.  In this case, $X_j$ and $X_k$ are conditionally independent given $X_i$, and hence $\tau_{kji} = 0$. On the other hand, by assumption, $\xi_{jk}\ge \delta$. Thus, $\tau_{kji} < \xi_{jk}$, and by $E$,
\[
\tau_{kji}^n \le \tau_{kji} + \frac{\delta}{2}= \frac{\delta}{2} = \xi_{jk} - \frac{\delta}{2}\le \xi_{jk}^n-\frac{\delta}{4} <\xi_{jk}^n.
\]
\item[Case 3.] $j\leftarrow i\leftarrow k$. The argument is identical to that of Case 2.
\item[Case 4.] $j\leftarrow i\to k$. Again, the same argument as in Case 2 works.
\end{enumerate}
This completes the proof of the lemma.
\end{proof}
Given Lemma \ref{approxlmm}, it is clear that the following lemma completes the proof of Theorem \ref{directionthm}.
\begin{lmm}\label{constructlmm}
The tree obtained by the modified algorithm is a valid DAG.
\end{lmm}
We need some preparation for the proof of Lemma \ref{constructlmm}.
\begin{lmm}\label{outgoing1}
Suppose that the dependency structure of $X = (X_i)_{i\in V}$ is described by an outgoing causal polytree on $V$. Then the dependency structure of $X$ is described by any other outgoing polytree on $V$ with the same skeleton.
\end{lmm}
\begin{proof}
Let $i\in V$ be the root, and $T$ be the skeleton, of an outgoing polytree that describes the dependency structure of $X$. Let $c(i)$ be the set of children of $i$ in this polytree, and let $U := \{i\}\cup c(i)$. Then the joint density of $X$ (with respect to a suitable product measure) is of the form
\begin{align*}
f_i(x_i) \biggl(\prod_{j\in c(i)} f_j(x_j|x_i) \biggr)\biggl(\prod_{k\in V\setminus U} f_k(x_k|(x_l)_{l\in p(k)}\biggr),
\end{align*}
where $p(k)$ denotes the set of parents of $k$ in  the polytree.
Take any $i'\in c(i)$. The above expression shows that the density of $X$ can be rewritten as 
\begin{align*}
f_{i'}(x_{i'}) f_i(x_i|x_{i'}) \biggl(\prod_{j\in c(i)\setminus\{i'\}} f_j(x_j|x_i) \biggr)\biggl(\prod_{k\in V\setminus U} f_k(x_k|(x_l)_{l\in p(k)}\biggr). 
\end{align*}
This shows that if we reverse the direction of the arrow from $i$ to $i'$, the resulting polytree is again a valid DAG for $X$.  But this is just the outgoing tree with skeleton $T$ and root $i'$. In other words, there is no harm is shifting the root to one of its neighbors. By induction, this completes the proof.
\end{proof}

\begin{lmm}\label{outgoing2}
Suppose that a polytree describing the dependency structure of $X= (X_i)_{i\in V}$ has a subtree that is outgoing, and every edge connecting a vertex in the subtree to a vertex outside the subtree is directed towards the vertex outside the subtree. Then, if we replace the directionalities of the edges in the subtree by any other set of directionalities that preserves the outgoing property, the resulting polytree is still a valid DAG for $X$.
\end{lmm}
\begin{proof}
Let $U$ be the vertex set of the subtree. Since each edge connecting a vertex of $U$ to a vertex outside $U$ is directed outwards, $p(i)\subseteq U$ for any $i\in U$. Since the joint density of $X$ can be written as
\begin{align}\label{twobr}
\biggl(\prod_{i\in U} f_i(x_i|(x_j)_{j\in p(i)})\biggr) \biggl(\prod_{i\not\in U} f_i(x_i|(x_j)_{j\in p(i)})\biggr),
\end{align}
the above observation implies that the joint density of $(X_i)_{i\in U}$ can be obtained by integrating out $(x_i)_{i\not\in U}$, to get
\[
\prod_{i\in U} f_i(x_i|(x_j)_{j\in p(i)}).
\]
In particular, the subtree describes the dependency structure of $(X_i)_{i\in U}$. Since the subtree is outgoing, it can be replaced by any other outgoing tree with the same skeleton, by Lemma~\ref{outgoing1}. Since the second product in \eqref{twobr} is the conditional density of $(X_i)_{i\not\in U}$ given $(X_i)_{i\in U}$, and $p(i)$ remains unchanged after the above replacement for any $i\not\in U$, this completes the proof of the lemma.
\end{proof}
We are now ready to prove Lemma \ref{constructlmm}, and hence, Theorem \ref{directionthm}.
\begin{proof}[Proof of Lemma \ref{constructlmm}]
Fix a valid set of directionalities for the edges of the skeleton $T$. A basic observation that we will use throughout the proof is the following. Let $i$ be a vertex and $j,k$ be two neighbors of $i$ in $T$. If $j\to i\leftarrow k$, then by one of the assumptions of Theorem \ref{directionthm}, $\tau_{kji} \ge \delta$, and by the nature of causal polytrees, $\xi_{jk}=0$. Therefore in this case, $\tau_{kji}\ge \xi_{jk}$. On the other hand, if $j\to i\to k$ or $j\leftarrow i \leftarrow k$ or $j\leftarrow i\to k$, then $X_j$ and $X_k$ are conditionally independent given $X_i$, and hence $\tau_{kji}=0$, and by one of the assumptions of Theorem \ref{directionthm}, $\xi_{jk}\ge \delta$. Therefore in this case, $\tau_{kji} < \xi_{jk}$. We will refer to this as the `basic observation' below.

We will say that a vertex is of type 1 if it has zero or one incoming edges, and of type 2 if it has at least two incoming edges. We will say that an edge is of type 1 if both of its endpoints are of type 1, and of type 2 if at least one endpoint is of type 2. 

We make two claims. The first claim is that upon completion of Steps 1 and 2 of the modified algorithm, there are no edges of $T$ whose directionalities have been incorrectly identified. This is proved by induction. Suppose that this holds up to a certain stage, at which we are inspecting a vertex $i$. There are two cases:
\begin{itemize}
\item[Case 1.] {\it No incoming edge into $i$ has yet been detected.} In this case, recall that for each pair of neighbors $(j,k)$ of $i$ in the estimated skeleton (sequentially in some order that is not dependent on the data), the algorithm checks  whether $\tau_{kji}\ge \xi_{jk}$. If this is found to be  true for some $(j,k)$, then the algorithm declares that the edges $(j,i)$ and $(k,i)$ are both directed towards $i$ (unless determined otherwise in a previous step). Suppose that this is found to be true for some $(j,k)$. If at least one of the edges $(j,i)$ and $(k,i)$ in the polytree is not directed towards $i$, then by the basic observation made above, $\tau_{kji}<\xi_{jk}$. Since we know that $\tau_{kji}\ge \xi_{jk}$, this is impossible. Thus, the edges from $j$ and $k$ to $i$ are both directed towards $i$ in the polytree. By induction hypothesis, the algorithm could not have determined otherwise in some previous step. Thus, the algorithm makes no mistake in this case.
\item[Case 2.] {\it The algorithm has already determined that some neighbor $j$ of $i$ in the skeleton has an edge directed towards $i$.} In this case, recall that the algorithm fixes some such $j$ and checks for every neighbor $k$ of $i$ (other than $j$) whether $\tau_{kji} \ge \xi_{jk}$. For those $k$ for which this holds, it declares that the edge $(k,i)$ points towards $i$ (unless determined otherwise in a previous step). For those $k$ for which this does not hold, it declares that the edge $(k,i)$ points towards $k$ (unless determined otherwise in a previous step). By the induction hypothesis, the edge from $j$ to $i$ is indeed directed towards $i$. By the basic observation made above, the edge from $i$ to a neighbor $k$ is directed towards $i$ if and only if $\tau_{kji}\ge \xi_{jk}$. Again by the induction hypothesis, no contradictory conclusion could have been reached in a previous step. Thus, the algorithm assigns directionalities correctly in this case too.
\end{itemize}
Our second claim is that upon completion of Steps 1 and 2 of the modified algorithm, the directionalities of all type 2 edges are correctly identified. To prove this claim, take a type 2 edge $e$. At least one of its endpoints is of type 2, meaning that it has two incoming edges. Call this endpoint $i$, and let $(j,i)$ and $(k,i)$ be any two incoming edges. Then by the basic observation made above (and the fact that the algorithm makes no mistakes), these incoming edges will be detected in a run of Step 1 of the modified algorithm. Once an incoming edge is identified, the directionalities of all edges incident to $i$ will be correctly identified in subsequent runs of Step~1, again by the basic observation. In particular, the directionality of $e$ will be correctly identified by the time Steps 1 and 2 are completed. 

Thus, by the end of Step 2, the only edges whose directionalities have not yet been identified must be edges of type 1 --- that is, both endpoints of type 1. We now claim that no errors are made in Step 3. To see this, take any vertex $i$ that is examined in Step 3 and found to have an incoming edge. If there is an edge incident to $i$ whose directionality has not yet been identified, it must be an edge of type 1. Thus, $i$ must be a vertex of type 1. Since it already has an incoming edge (because the algorithm has made no mistakes yet), all the remaining edges incident to $i$ must be outgoing. Thus, no mistakes are made in Step 3.

After completing Step 3, the only edges whose directionalities remain unidentified are type 1 edges. Taken together, these edges form a bunch of disjoint subtrees of $T$. Let $S$ be one of these subtrees. Let $U$ be the set of vertices of $S$. Then each element of $U$ is a type 1 vertex, that is, it can have at most one incoming edge. Thus, in the causal polytree we have chosen, the directionalities of the edges of $S$ makes it an outgoing tree. Moreover, since we have completed Step 3, no edge connecting a vertex of $S$ to a vertex outside $S$ can be going out of $S$. By Lemma \ref{outgoing2}, this implies that any reassignment of directionalities to the edges of $S$ that preserves the outgoing property gives a valid DAG. Since this is true for every such $S$, Step 4 now yields a valid DAG.
\end{proof}

\end{document}